\documentclass[12pt]{article}
\pdfoutput=1
\usepackage[T1]{fontenc}
\usepackage[utf8]{inputenc}
\usepackage{lmodern}
\usepackage{appendix}
\usepackage{hyphenat}
\usepackage{amsmath,amssymb,amsthm,mathtools}
\usepackage{mleftright}
\mleftright

\usepackage[margin=1.25in]{geometry}
\usepackage[onehalfspacing,nodisplayskipstretch]{setspace}
\usepackage[authoryear,round]{natbib}
\usepackage{booktabs,array}
\usepackage{caption}
\usepackage[colorlinks=true,citecolor=blue,urlcolor=blue,linkcolor=blue]{hyperref}
\usepackage[noabbrev]{cleveref}
\makeatletter
\AddToHook{begindocument/end}{\ifx\ltx@label\label\let\label@in@display\cref@old@label@in@display\fi}
\makeatother
\allowdisplaybreaks[2]
\hypersetup{pdftitle={A LATE-Targeting Experimental Design for Multi-Arm RCTs},pdfauthor={Tetsuya Kaji}}
\newtheorem{theorem}{Theorem}
\newtheorem{proposition}[theorem]{Proposition}
\newtheorem{lemma}[theorem]{Lemma}

\newtheorem*{claim}{Claim}
\theoremstyle{definition}
\newtheorem{definition}{Definition}[section]
\newtheorem{assumption}{Assumption}
\theoremstyle{remark}
\newtheorem{remark}{Remark}[section]
\crefname{assumption}{Assumption}{Assumptions}
\Crefname{assumption}{Assumption}{Assumptions}
\crefname{theorem}{Theorem}{Theorems}
\crefname{lemma}{Lemma}{Lemmas}
\crefname{proposition}{Proposition}{Propositions}
\crefname{definition}{Definition}{Definitions}
\crefname{section}{Section}{Sections}
\crefname{appendix}{Appendix}{Appendices}
\newcommand{\E}{\mathbb E}
\newcommand{\K}{\mathcal K}
\newcommand{\Rs}{\mathcal R}
\newcommand{\Ss}{\mathcal S}
\newcommand{\ind}[1]{\mathbf 1\{#1\}}
\newcommand{\row}{\operatorname{row}}
\DeclareMathOperator*{\argmax}{arg\,max}

\newcommand\z{\hphantom{0}}
\newcommand\cm{\hphantom{,}}

\title{A LATE-Targeting Experimental Design for Multi-Arm RCTs%
\thanks{This work is supported by the Richard N.\ Rosett Faculty Fellowship at the University of Chicago Booth School of Business.
I thank Andrew Jordan, Hidehiko Ichimura, Jack Mountjoy, Shosei Sakaguchi, Michal Koles\'ar, Edward Vytlacil, Azeem Shaikh, Peter Hull, Joshua Angrist, Vitalii Tubdenov, Kevin Chen, Sukjin Han, Lixiong Li, Seth Richards-Shubik, Max Tabord-Meehan, and seminar participants for helpful comments. Generative AI was used to formalize the proofs and polish the writing; all ideas were generated by the author.}}
\author{Tetsuya Kaji\\University of Chicago}
\date{\today}

\begin{document}
\maketitle
\begin{abstract}
We study the design of multi-arm experiments that identify local average treatment effects under imperfect compliance. With three or more arms, conventional randomization fails to identify margin-specific effects even under strong extensions of monotone compliance. A choice-theoretic analysis unifies these extensions and motivates a two-stage design: subjects report a preferred arm and are randomized conditional on it. If treatment is confined to the assigned or reported arm, one-sided noncompliance is restored conditional on the report, identifying the effect between the two arms for their compliers. Any tilt of assignment toward the reported arm makes reporting the favorite strictly optimal, and no assignment lottery can strictly elicit the ranking below the favorite. We then discuss the design's applicability and feasibility in relation to the choice model.
\end{abstract}

\section{Introduction}\label{sec:intro}

Economic experiments can be costly to organize and implement. A multi-arm randomized controlled trial (RCT) allows us to compare several interventions using the same infrastructure and control group. Researchers can then address several policy questions in one study, with the choice of arms determining which questions can be answered \citep{dgk2008}. Welfare-to-work experiments, for example, have compared service strategies and program components \citep{gh2023}.%

These comparisons become harder to interpret, however, when subjects do not receive their assigned treatments. Imperfect compliance arises in economic experiments when participants decline their assigned treatment or seek another, or when providers deviate from protocol. In a French counseling experiment, some job seekers entered a program other than the one assigned \citep{bcg2014}. In an Ethiopian experiment, agencies introduced programs in some communities assigned to other arms and failed to introduce them where assigned \citep{tdj2015}.

One form of imperfect compliance for which treatment effects are identified is \emph{one-sided noncompliance}. Individuals assigned to an active treatment may receive it or fall back to the control. Individuals assigned to the control must stay. Let \(K\) be the number of arms and \(R_i(z)\) the treatment unit \(i\) would receive if assigned to \(z\in\mathcal{K}\coloneqq\{1,\ldots,K\}\). Denoting the control by \(c\in\mathcal{K}\), one-sided noncompliance requires
\[
R_i(z)\in\{z,c\}\qquad\text{for every \(z\)}.
\]
Each active arm can then be compared with the control as in a binary experiment.\footnote{It identifies the effect on recipients of \(z\), known as the Bloom result \citep[Theorem~4.4.2, p.~164]{ap2009}.}
We can relax one-sided noncompliance by allowing arm-wise always-takers, who receive a fixed arm under \emph{every} assignment. This condition is called \emph{extended monotonicity} \citep{bcg2013}. Then every comparison of an active arm and the control reduces to monotone compliance in a binary experiment.

Even with this relaxation, however, extended monotonicity can be too restrictive in applications. \Cref{tab:tarozzi} illustrates the compliance in the experiment in \citet{tdj2015}, which assigned Peasant Associations to microfinance, family planning, both programs, or neither. The columns show the assigned programs and the rows show the programs implemented. Each cell counts pooled household survey observations. Under one-sided noncompliance, only the diagonal cells and the last row can be nonzero, since any deviation must lead to the control. Extended monotonicity also allows off-diagonal entries in the active-treatment rows, but only through arm-wise always-takers. Consequently, for each active arm, the population probability of receiving that arm must be the same under every off-target assignment.
In \Cref{tab:tarozzi}, however, microfinance-only implementation appears in communities assigned to neither program and to both programs, but not in those assigned to family planning alone. The family-planning and joint-program rows also have uneven off-diagonal entries, with some cells zero and others positive. These patterns cast doubt on extended monotonicity and motivate allowing for more general forms of deviation.%
\footnote{In the French counseling experiment of \citet{bcg2014}, \citet[Example~2.1]{bht2025} formally reject extended monotonicity based on a test that also uses outcomes.}

\begin{table}[t]
\centering
\caption{Compliance table in the experiment of \citet{tdj2015}.}\label{tab:tarozzi}
\small
\begin{tabular}{lccccc}
\toprule
\midrule
&\multicolumn{4}{c}{Assignment}&\\
\cmidrule(lr){2-5}
Treatment & MF & FP & Both & Neither & Total\\
\midrule
Microfinance (MF) &2,406&\z\cm\z\z0&\z\cm\z83&\z\cm302&\z2,791\\
Family planning (FP) &\z\cm\z\z0&1,763&\z\cm\z87&\z\cm\z58&\z1,908\\
Both &\z\cm374&1,217&2,753&\z\cm\z\z0&\z4,344\\
Neither &\z\cm416&\z\cm194&\z\cm258&2,764&\z3,632\\
\midrule
Total &3,196&3,174&3,181&3,124&12,675\\
\bottomrule
\end{tabular}
\begin{minipage}{0.96\textwidth}
\footnotesize\singlespacing
\emph{Notes:} Author's tabulation of the replication data for \citet{tdj2015}.
\end{minipage}
\end{table}

This raises our first question: can a conventional multi-arm RCT identify the average effect of one treatment against another for some fixed subgroup under more general forms of noncompliance? Our first main result gives a negative answer for several extensions in the literature. \cref{thm:impossibility} shows that with three or more arms and unrestricted potential outcomes, no nontrivial margin-specific group effect is identified even when ordered monotonicity of \citet{ai1995}, unordered monotonicity of \citet{hp2018}, and stability and strict stability (our terminology) of \citet[Assumption~2.2]{bhmsv2025} and \citet[Assumption~2.3]{bht2025} hold together. The problem is that the data do not let us recover the mean outcome for a common group under any two treatments, allowing construction of two populations with the same observable distribution but different effects for that group.
Instead, these papers concern \emph{across-arm} local average treatment effects (LATEs) \citep{ai1995,hp2018}, testable implications \citep{bht2025}, and partial identification \citep{bhmsv2025}.
Some applied papers also fall back to the \emph{intent-to-treat} (ITT) effects \citep{tdj2015,aadfh2015}.

The kind of noncompliance that violates extended monotonicity does not presuppose irrational or unpredictable behavior. It can arise from ordinary rational choice. This raises our second question: how much noncompliance can a reasonable model of choice generate?
We start from a model where a subject weighs the perceived benefit of each arm against a positive cost of deviation, and ask which compliance functions it can produce under our maintained assumptions of the stable unit treatment value assumption (SUTVA) and exclusion.
The induced compliance behavior, called \emph{acyclicity}, subsumes a broader range of noncompliance behavior than the four restrictions discussed above.
Each of the four can then be represented by additional restrictions on the choice model: ordered monotonicity when the cost of moving up decreases with the assignment, stability and strict stability when costs depend on the treatment alone, and unordered monotonicity when costs are common across subjects and benefits are suitably restricted.

This model points to an answer to our third question: can we design an experiment so that margin-specific effects for some subpopulations are identified?
We treat behavioral assumptions about compliance as \emph{inputs} to experimental design, and align the design with subjects' incentives so that their compliance behavior is visible and predictable to the administrator. Our protocol asks each subject to report a preferred treatment before randomization and randomizes conditional on the report. The key identifying restriction, \emph{compliance to stated preference}, requires a subject to receive either the assigned treatment or the stated preferred treatment.
Under this restriction, a Wald ratio identifies the effect of the reported arm relative to each alternative for subjects who report that arm but would accept the alternative. This estimand coincides with that of \citet[Proposition~2(iii)]{klm2016}, who condition on an applicant's next-best field; what an experiment adds is the ability to create this structure by design. Because subjects can switch after assignment, no assignment lottery can credibly elicit a next-best arm, but any tilt toward the report elicits the favorite, the one arm from which a subject has no reason to deviate further.
The intuition is that, conditional on the stated preferred arm, the compliance behavior is back to one-sided noncompliance where the preferred arm plays the role of the control.

The choice model supplies a reasonable protocol to enforce this restriction.
The survey report is used in two ways.
First, we can tilt the assignment probability toward the reported arm so that each subject is incentivized to report a true favorite.
Second, when a deviation must be accommodated, the administrator can reassign the subject to her reported favorite, so she has no incentive to deviate further.
Under the maintained choice assumptions the favorite is a fixed point: any further deviation would trade the highest benefit for a lower one at a positive cost.
The reassignment step can be omitted under strict stability with the reported favorite as the common deviation destination.

This paper relates to a wide range of subjects in causal inference.
While we focus on margin-specific treatment effects, many papers on multivalued treatment concern across-arm treatment effects, where the comparison is an average over several treatment margins.
Under monotonicity, \citet{ai1995} show that an instrumental variable (IV) with ordered treatment intensity identifies a weighted average of causal responses across treatment increments. For unordered treatments, \citet{hp2018} develop a framework whose identified effects can compare one treatment with different alternatives for different compliance types.
\citet{ls2024} likewise show that targeting instruments can identify across-arm effects.
\citet{ast2025} identify averages across treatment margins using a binary instrument and separate effects for switcher types under additional restrictions on mean potential outcomes. Our impossibility result focuses on \emph{margin-specific} effects and leaves potential outcomes unrestricted.

Other studies identify margin-specific effects by combining restrictions on treatment choice with additional information or richer instruments. \citet{klm2016} use admissions cutoffs and preference rankings to identify field-of-study LATEs, conditioning on the next-best alternative and imposing an irrelevance restriction on responses to the instrument. \citet{p2022} uses the incentives of Moving to Opportunity to restrict neighborhood choices and adds interpolation assumptions to separate effects of neighborhood types. \citet{m2022} separates community-college effects relative to no college from those relative to four-year college. Identification uses continuous distance instruments and unordered partial monotonicity, together with an assumption equating mean potential outcomes under community college for marginal compliers induced by different instruments. \citet{ls2018} identify margin-specific effects using multidimensional threshold-crossing models and sufficiently rich continuous instruments. Their estimands are marginal treatment effects (MTEs) indexed by latent selection variables, rather than the complier-group LATEs identified by our design. With discrete instruments, \citet{ls2024} obtain bounds on margin-specific LATEs under positive-selection restrictions, and point identification under equality restrictions on the relevant group mean outcomes.

A third line of work incorporates subjects' preferences into experimental design. \citet{r1989} combines free-choice and randomized-treatment arms in a two-stage trial. Subsequent work develops preference-specific estimands and IV interpretations of such designs \citep{lll2008,wc2017}. \citet{kybb2019} analyze multivalued preference trials, distinguishing stated preferences from actual choices and studying identification and sensitivity to discrepancies between them. \citet{k2023} uses report-dependent assignment probabilities to elicit truthful preferences and identify preference-specific effects in a binary setting. Mechanism-design approaches to experiments elicit subjects' valuations with incentives and use them in assignment: selective trials screen willingness to pay through menus of treatment probabilities at prices \citep{cps2012}, Experiment-as-Market computes assignment probabilities from reported willingness to pay in a pseudo-market with an approximate incentive guarantee \citep{narita2021}, and willingness-to-pay elicitation by the Becker--DeGroot--Marschak mechanism before a randomized price recovers effects by valuation in the field \citep{bfg2020}. These designs use money or cardinal reports and treat assignment as receipt.
Our design addresses imperfect compliance with multiple arms by combining a preference report with a restriction on final receipt.
\citet{tppz2015} study incentive-compatible experiments in which the strategic parties are the treatment providers rather than the subjects: providers choose which version of a treatment to administer in a winner-take-all comparison, and the scoring rule is designed so that the performance-maximizing version is a dominant strategy.

More broadly, our analysis relates to work on the role of assignment and selection in causal interpretation. \citet{goff2025} and \citet{nps2026} provide general criteria for IV identification. \citet{chen2026} characterizes the assignment processes under which a given linear regression estimates a contrast of potential outcomes without restricting those outcomes. \citet{bh2023} study treatment choices generated by a game of strategic substitutes and derive bounds on average treatment effects using excluded instruments and restrictions on selection and outcomes. Our focus is on using restrictions on compliance to design experiments for identification.

Importantly, several topics are kept outside the scope of this paper. We do not study across-arm LATEs, which average multiple treatment margins for a group identified by an instrument, or partial identification of margin-specific LATEs; both have their own literatures, discussed above and after \cref{thm:impossibility}.
We maintain the no-interference and exclusion assumptions throughout, while acknowledging that there is an established literature on interference \citep{hh2008,manski2013,ijm2021,vb2023} and a recent paper that casts doubt on exclusion in an experimental setup \citep[Section~8.1, p.~32]{h2026}.
While we set the target at margin-specific LATEs, we do not take a particular stand on their policy relevance; see \citet{hu2010}, \citet{deaton2010}, and \citet{imbens2010} for a debate and \citet{mt2024} for a synthesis on this matter.
These simplifications are deliberate.
We focus on making progress toward identification of margin-specific LATEs in a simple but difficult situation of multi-arm RCTs with general imperfect compliance.

The rest of the paper is organized as follows. \Cref{sec:setup} defines the experiment and the target effects. \Cref{sec:models} reviews existing extensions of monotone compliance, and \cref{sec:impossibility} gives the impossibility result. \Cref{sec:choice} develops the choice-theoretic foundations of compliance, and \cref{sec:design} presents the protocol and its identifying formula.
The appendices provide the proofs of all results.

\section{Multi-arm randomized experiments}\label{sec:setup}

We distinguish the treatment assigned to a unit from the treatment it receives. Let \(\K=\{1,\ldots,K\}\) be the index set of treatment arms. Unit \(i\) is assigned to \(Z_i\in\K\) and receives \(D_i\in\K\). A unit can be an individual or a randomized cluster. We study identification from population distributions.

Throughout the paper, we maintain the following assumption.

\begin{assumption}[SUTVA, exclusion, and consistency]\label{ass:outcomes}
Treatment choice depends only on the unit's own assignment.
This allows us to represent treatment receipt by a \emph{compliance function} of \(Z_i\),
\[
R_i:\K\longrightarrow\K,\qquad D_i=R_i(Z_i).
\]
The unit's outcome depends only on its own received treatment and not directly on assignment.
Thus we may write potential outcomes as \(\mathbf Y_i=(Y_{i1},\ldots,Y_{iK})\) such that
\[
D_i=R_i(Z_i),\qquad Y_i=Y_{iD_i}.
\]
\end{assumption}

SUTVA rules out interference across units in treatment receipt and outcomes.
Exclusion requires that holding received treatment fixed, the unit's own assignment does not affect its outcome.
Consistency assumes that the observed outcome equals the potential outcome under the assignment and treatment actually realized.

Let \(\Rs\) be the set of allowed compliance functions. Subjects who share a compliance function form a \emph{compliance type}; \citet{hp2018} call these response types. A restriction on \(\Rs\) specifies which responses to assignment are possible. For example, perfect compliance permits only the identity map \(R(z)=z\). \Cref{sec:models} reviews restrictions on \(\Rs\) that extend monotone compliance, and \cref{sec:choice} derives them from a model of treatment choice. Before turning to these restrictions, we define the treatment effects we seek to identify.

\subsection{Groups and margin-specific effects}\label{sec:groups}

To define a treatment effect, we must also specify whose outcomes we average. Let \(\mathcal I\) be the index set of subjects,\footnote{For a finite population, we may take \(\mathcal I=\{1,\ldots,n\}\). If we take a stand that there is a superpopulation, say for asymptotic arguments, we may view it as \(\mathcal I=\mathbb N\) or \(\mathcal I=[0,1]\).} and let \(X_i\) record their pre-assignment characteristics, including their compliance functions, preferences, and deviation costs. Let \(\Delta(\mathcal I)\) denote the probability simplex, that is, the set of probability distributions on \(\mathcal I\). A \emph{group} is an element \(G\in\Delta(\mathcal I)\). This definition includes ordinary subpopulations and allows their members to be represented in different proportions. Groups may be specified directly or by pre-assignment characteristics; subpopulations defined by potential outcomes are treated after \cref{def:target}.

We describe a population by the distribution \(\mu\in\Delta(\mathcal I)\) of its subjects together with the conditional distribution of each subject's potential outcomes \(\mathbf Y_i\), and write \(P\) for the joint law of \((X_i,\mathbf Y_i)\) that they induce. We write \(\E[\,\cdot\mid i]\) for expectation under the conditional distribution given subject \(i\). We allow arbitrary population distributions of pre-assignment characteristics and arbitrary conditional potential-outcome distributions, subject to the restrictions maintained for each experiment and the moment conditions below. Every compliance function in \(\Rs\) is represented among the subjects indexed by \(\mathcal I\), although a given population may assign it zero probability. We assume the expectations below exist.

\begin{samepage}
\begin{definition}[Margin-specific group ATE]\label{def:target}
For \(t\ne s\) and a group \(G\in\Delta(\mathcal I)\) with \(\int_{\mathcal I}\E[\lvert Y_{id}\rvert\mid i]\,dG(i)<\infty\), \(d\in\{t,s\}\), the \emph{margin-specific group average treatment effect (ATE)} is defined by
\begin{equation}\label{eq:group}
\tau_G(t,s)=\int_{\mathcal I}\E[Y_{it}-Y_{is}\mid i]\,dG(i).
\end{equation}
\end{definition}
\end{samepage}

For a subpopulation defined by pre-assignment characteristics and having positive probability, taking \(G\) to be the population distribution conditional on membership gives its conditional ATE. For a subpopulation defined by a fixed event \(A\) in \((X_i,\mathbf Y_i)\), the corresponding effect is the conditional ATE \(\E[Y_{it}-Y_{is}\mid(X_i,\mathbf Y_i)\in A]\). Taking \(G\) to be the whole population's distribution gives the population ATE. Effects for instrument-induced groups are called local average treatment effects (LATEs).
The definition naturally includes convex combinations of group effects.

The target in \eqref{eq:group} is \emph{margin-specific} because it compares the same two treatments for every subject, using the same group distribution for both. It differs from an \emph{across-arm} effect, which averages over different treatment margins. Ordered-treatment IV, for example, can average effects across successive treatment intensities \citep{ai1995}. Comparing a program with each participant's next-best alternative can also combine different margins \citep{kw2016,ast2025}. Such across-arm effects may be identified even when no margin-specific effect is identified, but they do not generally identify the target in \eqref{eq:group}.

For a fixed assignment rule, a target is identified if any two admissible laws generating the same distribution of the data recorded by that experiment give it the same value. To prove nonidentification, we construct two admissible populations with the same observable distribution and different target values.
To establish identification, we find a group for which the data identify the mean potential outcomes under two treatments.

\section{Existing extensions of monotone compliance}\label{sec:models}

With more than two arms, monotonicity can be extended in several ways.
This section reviews several extensions in the literature and compares them. In \cref{sec:impossibility}, we combine these restrictions to strengthen our impossibility result.

Throughout this section, the restrictions concern responses to assignment and leave measured potential outcomes unrestricted. \Cref{sec:choice} gives them choice-theoretic foundations, which help interpret the restrictions and the proposed design.

\subsection{Ordered monotone compliance}\label{sec:ordered}

When the arms have an order, one way to extend monotonicity is to require received treatment to be nondecreasing in assignment \citep{ai1995,v2006}. The restriction uses the order of the arms to limit the direction of responses.

\begin{definition}[Ordered monotone compliance] \label{def:ordered}
A compliance function is \emph{ordered monotone} if \(z<z'\) implies \(R(z)\le R(z')\).
\end{definition}

In our setting where the random assignment serves as the instrument, assignment and treatment share the same support and arm ordering, so it is natural that the admissible family contains the full-complier map \(R(z)=z\).
The possibility of full compliance fixes the direction of every pairwise comparison: the common direction can neither differ across pairs of assignment values nor be reversed by relabeling them.
On such families, ordered monotonicity is therefore equivalent to the pairwise common-direction restriction of \citet{ia1994}, extended to ordered treatments by \citet{ai1995} and stated as Assumption~A-2 of \citet{hp2018}, who also call it ordered monotonicity.

Our assignment is also a single scalar instrument, which removes the distinctions that arise with several instruments.
In the terminology of \citet[Sections~II.A and~II.D]{mtw2021}, interpreting their binary-treatment conditions in the ordered-treatment sense, actual monotonicity is ordered monotonicity restated, partial monotonicity coincides with Imbens--Angrist (IA) monotonicity because there is one instrument, and IA monotonicity coincides with ordered monotonicity by the previous paragraph.
The four notions therefore reduce to \Cref{def:ordered}.

\subsection{Stable and strictly stable compliance}\label{sec:stable}

A different restriction concerns the assignment being ``encouragement'' to take that arm, which can be defined without introducing an order on the arms.
We call compliance stable if a subject complies with an assigned arm that is her destination of noncompliance when assigned to some other arm.
Mathematically, we can also frame this as robustness to reassignment, although we stress that this notion is conceptually independent of whether reassignment is practically relevant.
Strict stability further requires every deviation by that subject to lead to the same treatment.

\begin{definition}[Stability]
A compliance function is \emph{stable} if \(R(R(z))=R(z)\) for every \(z\). It is \emph{strictly stable} if all nonfixed assignments have one common destination:
\[
R(z)\ne z,\ R(z')\ne z'
\quad\Longrightarrow\quad R(z)=R(z').
\]
\end{definition}

Trivially, strict stability implies stability. If a nonfixed assignment goes to \(q\), then \(q\) must be fixed. Otherwise, its deviation would also have to go to \(q\).

These restrictions also connect to existing generalizations of monotonicity. If a subject receives \(d\) under any assignment, stability requires her to receive \(d\) when assigned to \(d\). It therefore satisfies \emph{generalized monotonicity} in \citet[Assumption~2.2]{bhmsv2025}, with each arm's own assignment as its maximally encouraging value \citep[Example~4.1]{bhmsv2025}.
Strict stability requires all deviations by a subject to go to one treatment, as in the default-treatment restriction of \citet[Assumption~2.3]{bht2025}.

Stability is also the monotonicity condition of \citet[Assumption~1d, p.~827]{hhklm2026}, who require that anyone choosing field $k$ under another offer also chooses $k$ when offered $k$; with three fields their admitted compliance types (\citealp[Table~A1]{hhklm2026}, online appendix) correspond exactly to the ten stable compliance functions.%
\footnote{The weaker condition of \citet[Assumption~4, p.~1065]{klm2016} restricts only responses to the reference offer.}

\subsection{Unordered monotone compliance}\label{sec:unordered}

The restrictions so far describe how an individual responds to different assignments. Unordered monotonicity instead compares responses across individuals by restricting a collection of compliance functions \citep{hp2018}. For each treatment \(d\), define its assignment basin
\[
A_d(R)\coloneqq\{z\in\K:R(z)=d\}.
\]

\begin{definition}[Unordered monotonicity]
A collection \(\Rs\) is \emph{unordered monotone} if, for every \(d\), the sets \(\{A_d(R):R\in\Rs\}\) form a chain under inclusion.
\end{definition}

Inclusion is weak, and the ordering of compliance functions may differ by destination. Equivalently, for every treatment \(d\) and pair of assignments \(z,z'\), the change from \(\ind{R(z)=d}\) to \(\ind{R(z')=d}\) has the same weak sign for every compliance type. Incomparable basins would give two types with opposite responses to this change.

To relate this restriction to an experimental setup, we also require that the collection contain at least the compliance function of a full complier, the identity map \(I\). Without this requirement, a singleton collection containing any compliance function is unordered monotone. With the identity included, every nonempty basin \(A_d(R)\) must contain \(d\), since it must be comparable with \(A_d(I)=\{d\}\). Every compliance function in the collection is therefore stable.

\subsection{Comparing the restrictions}\label{sec:comparison}

The above restrictions limit compliance in different ways. Neither ordered monotonicity nor stability implies the other. Strict stability is stronger than stability, and unordered monotonicity of a set with the identity implies stability of every member thereof. \Cref{tab:responses} illustrates the pointwise distinctions.

\begin{table}[t]
\centering
\caption{Examples of individual compliance functions.}\label{tab:responses}
\small
\begin{tabular}{cccc}
\toprule
\midrule
Compliance function & Ordered & Stable & Strictly stable\\
\midrule
\((1,2,3)\)   & Yes & Yes & Yes\\
\((2,3,3)\)   & Yes & No  & No\\
\((3,2,3)\)   & No  & Yes & Yes\\
\((1,1,3,3)\) & Yes & Yes & No\\
\bottomrule
\end{tabular}
\begin{minipage}{0.96\textwidth}
\footnotesize\singlespacing
\emph{Notes:} The \(z\)th entry is \(R(z)\). The fourth row has four arms; the others have three. Unordered monotonicity is a property of a collection and is therefore not assigned a column.
\end{minipage}
\end{table}

Unordered monotonicity adds a restriction even when every individual function is strictly stable and ordered monotone. For example, the functions \((2,2,3)\) and \((1,2,2)\) each have these properties, but cannot belong to the same unordered-monotone collection. Their treatment-2 basins, \(\{1,2\}\) and \(\{2,3\}\), are not nested. Conversely, the identity together with \((1,1,3,3)\) is unordered monotone, although the second map is not strictly stable.

For the impossibility result, we combine these restrictions.
This is because starting from a stronger restriction makes an impossibility result more powerful.
The impossibility will then apply to any weaker compliance assumptions, such as when only ordered monotonicity is imposed.
To be precise, we use inclusion-maximal unordered-monotone collections within the class of strictly stable, ordered monotone compliance functions. A collection is \emph{inclusion-maximal} if no further function from that class can be added without violating unordered monotonicity.
This specifies the possible types that cannot be excluded from consideration, but importantly, it does not mean that all types in the set must exist in a particular population.
The next section shows that even with these joint restrictions, margin-specific effects are not identified.

\section{Impossibility in a conventional experiment}\label{sec:impossibility}

We first consider a conventional RCT in which the researcher observes \((Y_i,D_i,Z_i)\). Assignment is randomized unconditionally, as specified in the following assumption.

\begin{assumption}[Random assignment]\label{ass:random}
\[
Z_i\ \perp\ (X_i,\mathbf Y_i),\qquad
P(Z_i=z)>0\quad\text{for every }z\in\K.
\]
\end{assumption}

The compliance restrictions in the previous section limit how subjects can respond, but they do not solve the problem of comparing the same group under two treatments. Random assignment makes assignment groups comparable, whereas the groups receiving different treatments may differ. We first show that, for the classes of compliance functions considered here, assignment-cell contrasts cannot recover the mean outcome for any common nonempty group under two treatments. We then connect this failure to two populations with identical observables and different group effects, which rules out identification by other methods as well.

\subsection{Recovering the same group under two treatments}

To determine which groups we can recover, consider weighted sums of observed outcome moments. For any coefficient vector \(b\in\mathbb R^K\), random assignment and exclusion imply
\[
\sum_{z\in\K}b_z\E[Y_i\ind{D_i=t}\mid Z_i=z]
=\E\left[Y_{it}\sum_{z:R_i(z)=t}b_z\right].
\]
Each coefficient vector weights \(Y_{it}\) by compliance type. To compare \(t\) and \(s\) for the same group, the weights must agree:
\begin{equation}\label{eq:alignment}
\sum_{z:R(z)=t}b_z=\sum_{z:R(z)=s}c_z
\quad\text{for every allowed }R.
\end{equation}
The binary LATE contrasts have the complier indicator as their common weight. In the class below, only zero coefficients satisfy \eqref{eq:alignment}.

\begin{lemma}[No nonzero group alignment]\label[lemma]{lem:alignment}
Let \(K\ge3\). Suppose \(\Ss\) is inclusion-maximal unordered monotone in the class of strictly stable, ordered monotone compliance functions. For every pair \(t\ne s\), real vectors \(b\) and \(c\) satisfying \eqref{eq:alignment} for all \(R\in\Ss\) must have \(b=c=0\).
\end{lemma}

The identity belongs to every such maximal collection. Its basin for treatment \(d\) is \(\{d\}\), which is comparable with every basin of a stable compliance function. Adding the identity can therefore never violate the restrictions.

The other compliance functions in this class collapse an interval of assignments to one destination \(q\) and leave the remaining assignments fixed:
\begin{equation}\label{eq:collapse}
M_{q,A}(z)=
\begin{cases}
q,&z\in A,\\
z,&z\notin A,
\end{cases}
\qquad A=[\ell,r],\quad \ell\le q\le r.
\end{equation}
For each destination, unordered monotonicity requires these intervals to form a chain. Maximality requires the chain to start at \(\{q\}\), end at \(\K\), and add one assignment at each step. The chains for different destinations can be chosen separately because, under a collapse in \eqref{eq:collapse}, any treatment other than \(q\) has either an empty basin or its singleton basin.

Fix the two treatments \(t,s\). Successive differences along the \(t\)-chain and the \(s\)-chain force
\begin{equation}\label{eq:binary-form}
b=\alpha(e_t-e_s),\qquad c=\alpha(e_s-e_t).
\end{equation}
Here \(e_d\) denotes the \(d\)th coordinate vector in \(\mathbb R^K\). These are the binary contrasts. The chain for a third destination contains an interval with exactly one of \(t,s\). Collapsing that interval empties one treatment's basin and leaves the other's singleton basin, forcing \(\alpha=0\). The argument applies to every pair without relabeling the ordered arms. \Cref{app:impossibility} gives the details.

For example, with three arms every maximal collection contains
\[
I,\ (1,1,3),\ (1,1,1),\ (1,3,3),\ (3,3,3),\ (2,2,2),
\]
and exactly one of \((2,2,3)\) and \((1,2,2)\). For the pair \((1,2)\), the map \((1,3,3)\) has treatment-1 basin \(\{1\}\) and an empty treatment-2 basin. In words, a subject of this type behaves like a full complier when assigned 1 but never takes 2, so any contrast weights her \(Y_{i1}\) exactly as it weights a full complier's while giving her \(Y_{i2}\) no weight; alignment then forces that common weight to zero.
When \(K=2\), there is no third destination, and the nonzero complier contrast remains.

\subsection{From group alignment to nonidentification}

So far, the argument has ruled out assignment-cell contrasts that recover a common group. To rule out identification by any method, we need to connect this failure to observationally indistinguishable populations. With outcomes unrestricted, failure of group alignment allows us to change potential outcomes without changing observables. For groups specified without potential outcomes, the appendix constructs two populations with the same pre-assignment characteristics and the same distribution of \((Y,D,Z)\), but different ATEs for the same group \(G\). Since the data cannot distinguish these populations, they cannot determine the group effect.

This connection between group alignment and identification also appears in general identification criteria. \citet[Proposition~1 and Theorem~2]{goff2025} shows that identifying an effect for a compliance-type group requires its weights to lie in both treatment-specific row spaces. \citet{nps2026} give more general operator and weighting characterizations. In particular, Remark~4.3 covers ratios for conditional means. We give a direct proof using the finite set of compliance functions and allow groups to depend on other latent characteristics.

To include outcome-defined subpopulations, we must exclude effects fixed by their definitions. Call a group effect \emph{nontrivial} if it can take different values across admissible populations in which it is defined. For example, the effect for subjects satisfying \(Y_{it}=Y_{is}\) is always zero.

\begin{theorem}[Impossibility of margin-specific group ATE identification]\label{thm:impossibility}
Maintain \cref{ass:outcomes,ass:random}. Suppose \(K\ge3\) and \(\Rs\) contains a collection \(\Ss\) satisfying the restrictions of \cref{lem:alignment}. Then no nontrivial margin-specific group ATE is identified, either for groups specified without potential outcomes or for subpopulations defined by a fixed measurable criterion in \((X_i,\mathbf Y_i)\).
\end{theorem}

\begin{remark}[Linear combinations]\label{rem:signed-groups}
The proof in \cref{app:impossibility} also rules out identification of linear combinations with negative coefficients. For fixed groups \(G_\ell\) specified without potential outcomes, real coefficients \(a_\ell\), and a common treatment margin \((t,s)\), the contrast \(\sum_{\ell=1}^m a_\ell\tau_{G_\ell}(t,s)\) is unidentified whenever \(\sum_{\ell=1}^m a_\ell G_\ell\ne0\). This excludes combinations that cancel identically.
\end{remark}

For groups specified without potential outcomes, the theorem extends the alignment argument beyond compliance-defined groups. A group may be defined by preferences, costs, or other pre-assignment characteristics. The proof first chooses a population containing every allowed compliance type and representing the target group. It then changes potential outcomes while keeping population composition and the group fixed.
The theorem also applies to any larger compliance class containing \(\Ss\).
For example, the theorem applies to the ordered monotone class.

Stronger restrictions on compliance may escape the nonidentification result.
One-sided noncompliance and extended monotonicity remove the troublesome compliance functions used in the proof and can identify some ATEs, as discussed in the introduction.
Identification may also hold under additional restrictions on potential outcomes, as in \citet[Proposition~10(3)]{ls2024}.

It is also useful to distinguish this failure of identification from a failure of a particular estimator or a change in the target effect. \citet{bs2024} characterize when multi-treatment two-stage least squares (2SLS) coefficients put nonnegative weights on their intended treatment contrasts and zero weights on other contrasts. Their examples show that other methods can identify effects when 2SLS lacks this interpretation.
Our theorem differs in two ways. First, it is a general identification problem and applies to any choice of the estimator. Second, it concerns \emph{margin-specific} group effects. Switcher averages that combine different margins, as in \citet{ast2025}, are \emph{across-arm} effects and fall outside the scope of \Cref{thm:impossibility}.

Finally, we emphasize that \Cref{thm:impossibility} concerns \emph{point} identification; it does not state that the data are uninformative for \emph{partial} identification of treatment effects.
In fact, there is a literature on partial identification of treatment effects for groups defined via principal strata.
\citet{cs2006} derive bounds on causal effects within principal strata for three-arm trials with noncompliance.
\citet{bhmsv2025} characterize the identified set for ATEs under stability and under any stronger restriction that leaves outcomes unrestricted, so the population ATE is already not point identified under the classes covered by \cref{thm:impossibility}. Our result extends this to every group, including the complier groups that LATE arguments target.

\section{Choice-theoretic foundations of compliance}\label{sec:choice}

\Cref{thm:impossibility} shows that even a strong combination of existing restrictions does not restore identification in a conventional experiment. Noncompliance of this kind that goes beyond extended monotonicity, however, does not presuppose irrational or unpredictable behavior.
Even a mildly restrictive model of rational choice can produce enough noncompliance to ruin identification in \Cref{thm:impossibility}.
Such rationality, however, motivates integrating some incentive mechanisms to control the deviation directions.

With that in mind, we first examine how much noncompliance a reasonable model of choice can generate.
We begin with a choice model that is deliberately kept general, while respecting the structure compatible with SUTVA and the exclusion restriction.
This is intended to exhaust the compliance types that can arise from such choices.
We then show how the extensions discussed earlier in \cref{sec:models} can be reproduced by additional structural assumptions on this model.

\subsection{Acyclic compliance}\label{sec:acyclic}

Let \(B_{id}\) be subject \(i\)'s perceived benefit from treatment \(d\). If assigned to \(z\), she pays a deviation cost \(C_{izd}\) to receive another treatment \(d\). Taking the assigned treatment costs nothing, and every deviation incurs a positive cost:
\begin{equation}\label{eq:choice-model}
R_i(z)\in\argmax_{d\in\K}B_{id}-C_{izd},
\qquad C_{izd}>C_{izz}=0\quad(d\ne z).
\end{equation}
Benefits and costs may differ across subjects. Ties may be resolved by any deterministic, individual-specific rule, which may depend on the assignment.

The model describes treatment choice, so its benefits need not coincide with the outcomes the researcher studies. We do not impose structural restrictions on the relationship between the subject's perceived payoff, \(B_{id}\), and the outcome measured in the experiment, \(Y_{id}\).
In cases where \(Y_{id}\) measures outcomes related to the subject's benefits, \(B_{id}\) and \(Y_{id}\) may be strongly associated; in other cases, they may be completely unassociated.
Equation \eqref{eq:choice-model} assumes that perceived benefits depend on treatment but not assignment. This assumption naturally aligns with SUTVA and exclusion, although they concern different aspects of the experiment.

Allowing deviation costs to depend on assignment accommodates differences in how difficult a reassignment is to obtain. Consider, for example, arms offering cash transfers of \(\$0\), \(\$500\), and \(\$1{,}000\). Increasing a transfer from \(\$0\) to \(\$500\), or from \(\$500\) to \(\$1{,}000\), costs the administrator an additional \(\$500\). Moving from \(\$0\) to \(\$1{,}000\) costs an additional \(\$1{,}000\). The larger expense may make reassignment harder for the subject to obtain. The cost in \eqref{eq:choice-model} is the subject's cost of obtaining it, which need not equal the administrator's expense, but an unavailable reassignment can be represented by a cost high enough that it is never chosen.

Costs can also depend on assignment because assignment can change how a subject views a request for another treatment. She may dislike appearing greedy, and the same request may seem more appropriate under one assignment than another. Procedures and physical access may also depend on the assigned arm. Studies of context\hyp{}dependent choices \citep{kw2013,l2007,b2008} and the social cost of declining requests \citep{dlm2012} motivate allowing assignment\hyp{}dependent \(C_{izd}\).

Despite this flexibility in costs, the model rules out cycles of deviation. A subject deviates only if the other treatment offers a higher benefit:
\begin{equation}\label{eq:benefit-increase}
R_i(z)\ne z\quad\Longrightarrow\quad B_{i,R_i(z)}>B_{iz}.
\end{equation}
Since every deviation raises the subjective benefit, hypothetical repeated deviations cannot return to the starting arm.

\begin{definition}[Acyclic compliance]\label{def:acyclic}
A compliance function \(R\) is \emph{acyclic} if it has no cycle of length greater than one. Thus \(R^m(z)=R\circ\cdots\circ R(z)=z\) for some positive integer \(m\) implies \(R(z)=z\).
\end{definition}

Acyclicity still allows fixed points, which are the destinations of every chain of deviations. The following theorem shows that this gives a full characterization of compliance functions generated by the above choice model.

\begin{theorem}[Acyclic representation]\label{thm:acyclic}
A compliance function can arise from \eqref{eq:choice-model} if and only if it is acyclic. More precisely, arbitrary tie-breaking among maximizers in \eqref{eq:choice-model} yields an acyclic compliance function, and every acyclic compliance function admits a representation with a unique maximizer at each assignment in \eqref{eq:choice-model}.
\end{theorem}

The necessity follows from \eqref{eq:benefit-increase}. For sufficiency, assign higher benefits to arms closer to a fixed point, then choose costs that make each prescribed move optimal. The construction is in \cref{app:representations}.

\Cref{thm:acyclic} can be read as a revealed-preference result. Because every deviation is costly, a move from \(z\) to \(R(z)\) reveals \(B_{R(z)}>B_z\), as in \eqref{eq:benefit-increase}, whereas compliance reveals nothing, since deviation costs may be arbitrarily large. The classical conditions for choices to be rationalized by a preference ordering rule out cycles of revealed preference that contain a strict comparison (\citealp[p.~163]{h1950}; \citealp[Theorem~1, p.~639]{r1966}; \citealp[p.~74]{a1967}). Here every revealed comparison is strict, so the analogous condition rules out any cycle, which is acyclicity. Because costs may depend on assignment, nothing else ties together choices under different assignments, so any acyclic pattern of revealed comparisons is rationalized.

Acyclicity is therefore the weakest restriction on compliance that retains this rational-choice interpretation. It extends binary monotonicity, but it accommodates a broad range of behavior when there are more arms. With \(K=2\), it excludes precisely the defier. With more arms, it allows sequential deviation. For example, \(R=(2,3,3)\) describes a subject who would move from assignment 1 to treatment 2, but when assigned to 2, move to 3. She may choose this way because moving directly from 1 to 3 costs more than moving from 2 to 3.
This consideration is hypothetical, so it does not mean that the experiment must allow subjects to actually make these moves in sequence.
It is interesting to note that, while the formulation in \eqref{eq:choice-model} seems to encourage subjects to take the assignment by offering the lowest cost, it does not exclude sequential deviation that is ruled out in the encouragement design (stability).

\subsection{Representing existing extensions}\label{sec:costs}

Each extension in \cref{sec:models} corresponds to additional restrictions on the choice model in \eqref{eq:choice-model}.

\subsubsection{Ordered monotone compliance}

In the choice model, we obtain ordered monotone compliance (\cref{sec:ordered}) by requiring deviation costs to change systematically with assignment.
Namely, ordered monotonicity can be embedded with the strict submodularity of the cost matrix:
\begin{equation}\label{eq:monge}
C_{izd}+C_{iz'd'}<C_{izd'}+C_{iz'd}
\quad\text{whenever }z<z',\ d<d'.
\end{equation}
The incremental cost of moving to a higher treatment is lower at a higher assignment.
Counterintuitively, monotonicity of the perceived benefits is neither necessary nor sufficient.

\begin{theorem}[Ordered representation]\label{thm:ordered}
A compliance function can arise from \eqref{eq:choice-model} with a cost matrix satisfying \eqref{eq:monge} if and only if it is ordered monotone. More precisely, arbitrary tie-breaking among maximizers in \eqref{eq:choice-model} with \eqref{eq:monge} yields an ordered monotone compliance function, and every ordered monotone compliance function admits such a representation with a unique maximizer at each assignment in \eqref{eq:choice-model} with \eqref{eq:monge}.
\end{theorem}

The proof uses increasing differences in net payoffs. For the converse, it constructs costs satisfying both strict submodularity and the positive-cost normalization in \eqref{eq:choice-model}.

\citet{v2006} shows that ordered monotonicity is equivalent to an ordered discrete-choice model with random thresholds.
It is a population-level equivalence that maintains instrument independence, whereas \cref{thm:ordered} is an individual-level representation with no distributional assumption.

\subsubsection{Stable and strictly stable compliance} \label{sec:costs:stable}

In the choice model, stability (\cref{sec:stable}) follows from removing the assignment dependence that allowed sequential deviation in \cref{sec:acyclic}. Let deviation costs depend on the destination but not on the assignment:
\begin{equation}\label{eq:stable-model}
R_i(z)\in\argmax_{d\in\K}
B_{id}-C_{id}\ind{d\ne z},
\qquad C_{id}>0.
\end{equation}
Taking the assigned treatment costs nothing. The cost of receiving any other treatment depends only on that treatment.

\begin{theorem}[Stable representations]\label{thm:stable}
A compliance function can arise from \eqref{eq:stable-model} with arbitrary tie-breaking if and only if it is stable.
It can arise from \eqref{eq:stable-model} with a unique maximizer at every assignment if and only if it is strictly stable.
\end{theorem}

Write \(v_d=B_d-C_d\). Assignment to \(z\) raises only \(v_z\), by \(C_z>0\). A subject who deviates must choose a treatment maximizing the unchanged values \(v_d\). Each such treatment is a fixed point. Ties allow different assignments to lead to different fixed points, whereas unique choices require all deviations to go to one destination.

The utility in \eqref{eq:stable-model} has the encouragement form studied by \citet[Appendix~A, Lemma~A.1]{bht2025}. Baseline utility is \(B_{id}-C_{id}\), and assignment adds \(C_{id}\) only to the assigned arm. When each arm has a corresponding encouragement, they establish equivalence with strict stability under instrument independence, positive assignment probabilities, and an absolutely continuous distribution of baseline utilities. Their continuity assumption excludes baseline ties.

\subsubsection{Unordered monotone compliance}

Because unordered monotonicity restricts a collection of compliance functions (\cref{sec:unordered}), a choice-theoretic representation must restrict which compliance types can occur together. We represent these collections using common deviation costs and a restricted set of benefit profiles. Let \(C\) satisfy \(C_{zd}>C_{zz}=0\) for \(z\ne d\), and define
\[
T_d(B)=\{z:B_d>C_{zd}\}.
\]
Call \(B\) \emph{admissible} if the sets \(T_d(B)\), \(d\in\K\), partition \(\K\). Let \(\mathcal B\) be a collection of admissible profiles containing a profile \(B^I\) with
\[
\min_{z\ne d}C_{zd}\ge B_d^I>0\quad\text{for every }d.
\]
For each \(B\in\mathcal B\), choose \(R_B(z)\) to maximize \(B_d-C_{zd}\). Admissibility makes exactly one net payoff positive at each assignment, so this maximizer is unique. The profile \(B^I\) generates the identity.

\begin{theorem}[Unordered representation]\label{thm:unordered}
A collection of compliance functions is generated by some pair \((C,\mathcal B)\) as above if and only if it is unordered monotone and contains the identity.
\end{theorem}

Common costs alone, however, do not ensure unordered monotonicity because an arbitrary benefit distribution need not produce it. Restricting \(\mathcal B\) as above makes the assignment basins nested threshold sets. This gives a representation related to the threshold and response-matrix characterizations in \citet[Theorem~T-3]{hp2018}.

These representations also place the restrictions of \cref{sec:models} inside the acyclic class. Since each of these models is a special case of \eqref{eq:choice-model}, \cref{thm:acyclic} implies that every compliance function in these classes is acyclic. Every compliance function in \Cref{tab:responses} is therefore acyclic, whereas \((2,3,1)\) is not. Since the acyclic class contains every collection \(\Ss\) described in \cref{lem:alignment}, \cref{thm:impossibility} applies to the unrestricted choice model as well.

\section{A LATE-targeting experimental design}\label{sec:design}

If compliance reflects rational choice, the experiment can use subjects' own incentives to reveal and restrict where they deviate.

In the econometric terminology, the key idea is to \emph{collect a pre-assignment covariate conditional on which one-sided noncompliance is restored}. This is related to, but distinct from, IV models in which instrument exogeneity holds conditional on observed covariates \citep{a2003,t2006,t2006c,f2007,orb2015}.
In these models, conditioning addresses confounding of the instrument. For us, exogeneity is already supplied by randomization, while conditioning acts as a mechanism to control the compliance behavior within the conditioned stratum.

We then discuss how the choice model developed in \cref{sec:choice} provides guidance on its implementation and the limitations of this approach.

\subsection{The two-stage experiment}\label{sec:protocol}

In situations where subjects deviate based on payoff considerations, the information on their preference helps elicit credible compliance behavior.
For this, we collect information on the most preferred arm of each subject.

Namely, we add a survey stage before randomization.
This can, for example, be done when obtaining subjects' consent to participate in the experiment.
Given the report, we randomize assignment possibly conditional on each subject's reported preferred arm.
In particular, the two-stage experiment goes as follows.

\begin{enumerate}
\item Inform each subject of the collection of treatment arms \(\mathcal{K}\) and the randomization scheme. Ask each subject to state a preferred arm \(S_i\in\K\).
\item Randomize assignment \(Z_i\) possibly conditional on the reported \(S_i\). For example, fix a baseline distribution \(a\in\Delta(\K)\) with \(a_z>0\) for every \(z\), independently of the report, and for \(0\le\lambda<1\), set
\begin{equation}\label{eq:lottery}
P(Z_i=z\mid S_i=s)
=(1-\lambda)a_z+\lambda\,\ind{z=s}.
\end{equation}
If \(a_z\equiv 1/K\), then the assignment probability to non-chosen arms is uniform.
\item Observe \((Y_i,D_i,Z_i,S_i)\) for each subject.
\end{enumerate}

The following two assumptions are essential in exploiting this new scheme.

\begin{assumption}[Conditional random assignment]\label{ass:conditional}
Conditional on \(S_i\), assignment is independent of potential outcomes, the compliance function, and remaining pre-assignment characteristics. The probabilities \(P(Z_i=z\mid S_i)\) are positive for every arm \(z\).
\end{assumption}

\begin{assumption}[Compliance to stated preference]\label{ass:preference}
Let \(R_i(z)\) continue to denote the treatment choice of individual \(i\) under the new scheme.
For every initial assignment \(z\),
\[
R_i(z)\in\{z,S_i\}.
\]
\end{assumption}

This assumption implies that the observed treatments satisfy \(D_i\in\{Z_i,S_i\}\). This implication is empirically testable, making the assumption partly falsifiable.
Restrictions like \cref{ass:preference} have analogues in observational studies.
\citet{klm2016} estimate the earnings payoffs to fields of study, using discontinuities at admission cutoffs as instruments for the field an applicant is offered, together with the strategy-proof application rankings that record each applicant's preferred field and next-best alternative.
The condition behind their Proposition~2(iii), also restated in \citet[Assumption~2, p.~829]{hhklm2026}, requires that an offer that does not induce a person to take the offered field leaves her choice among the other fields unchanged, and a person without such an offer takes the stated next-best alternative.
Together these restrictions say that receipt is either the offered field or the next-best field, which is \cref{ass:preference} with \(S_i\) equal to the next-best field.
There the elicited arm is a fallback rather than a favorite, and an offer plays the role of \(z\) rather than a randomized assignment.

\subsection{Identification for preference-anchored groups}\label{sec:possibility}

Conditional on \(S_i=d\), compliance to stated preference amounts to one-sided noncompliance, with \(d\) serving as the control arm. Subjects may differ arbitrarily in which other assignments they accept, but any deviation must lead to \(d\). This is the common destination needed for the comparison in identification.

We can now apply the one-sided comparison within a reported-preference group. For \(z\ne d\), the target is
\[
\tau_{d,z}
=\E[Y_{id}-Y_{iz}\mid S_i=d,R_i(z)=z].
\]
It compares the reported preferred arm \(d\) with \(z\) for subjects who report \(d\) but would comply with assignment to \(z\).

\begin{theorem}[Identification under compliance to stated preference]\label{thm:possibility}
Maintain \cref{ass:outcomes,ass:conditional,ass:preference}. If
\(P(S_i=d,R_i(z)=z)>0\), then
\begin{equation}\label{eq:wald-ratio}
\tau_{d,z}
=
\frac{\E[Y_i\mid S_i=d,Z_i=d]-\E[Y_i\mid S_i=d,Z_i=z]}
{P(D_i=z\mid S_i=d,Z_i=z)}.
\end{equation}
\end{theorem}

Conditional on \(S_i=d\), subjects who decline \(z\) receive \(d\) under either assignment, canceling their contributions from the numerator. The formal proof is in \cref{app:design}.
\Cref{thm:possibility} is akin to \citet[Proposition~2(iii)]{klm2016}, which shows that, conditional on the next-best field, the IV estimand identifies the payoff to the preferred field relative to the next-best field for applicants induced into the preferred field by the offer.
Denoting the next-best field by \(d\) and the offered field by \(z\), their 2SLS coefficient reduces to the negative of the Wald ratio in \eqref{eq:wald-ratio}; it equals \(-\tau_{d,z}=\E[Y_{iz}-Y_{id}\mid S_i=d,\ R_i(z)=z]\). The two results therefore identify the same estimand, the effect for those who take the offered arm when it is assigned and otherwise take the elicited arm. The difference lies in how the elicited arm is used. It is the fallback in their design and the favorite in ours, so the group recovered there consists of applicants who obtain their preferred field, whereas the group recovered here consists of subjects who accept an alternative to their reported favorite.

We can also recover the two potential-outcome means separately:
\begin{align*}
\E[Y_{id}\ind{R_i(z)=z}\mid S_i=d]
&=\E[Y_i\mid S_i=d,Z_i=d]\\
&\hspace{50pt}-\E[Y_i\ind{D_i=d}\mid S_i=d,Z_i=z],\\
\E[Y_{iz}\ind{R_i(z)=z}\mid S_i=d]
&=\E[Y_i\ind{D_i=z}\mid S_i=d,Z_i=z].
\end{align*}
Divide each expression by \(P(R_i(z)=z\mid S_i=d)\) to obtain the two means. The receipt probabilities must be included because subtracting unweighted means from different received-treatment cells would compare different groups.

Once the identifying assumptions hold, estimation uses the sample counterpart of \eqref{eq:wald-ratio}. Subtract the two report-by-assignment sample means and divide by the fraction receiving \(z\) among subjects with \(S_i=d\) and \(Z_i=z\).
If assignment also depends on recorded baseline strata \(W_i\), standardize the cell means or use inverse-probability weighting with the known positive assignment probabilities. For example,
\[
\E\left[
\frac{\ind{Z_i=z}Y_i}{P(Z_i=z\mid S_i,W_i)}
\ \middle|\ S_i=d\right]
\]
recovers the mean outcome under assignment \(z\) in report group \(d\), under conditional randomization given the report and baseline strata. Weight the denominator in \eqref{eq:wald-ratio} in the same way. The resulting ratio keeps the intended group weights when assignment probabilities vary across strata.

The identification theorem itself does not require truthful reports, but a preference interpretation of \cref{ass:preference} does.
We now ask how \cref{ass:preference} can be implemented in the context of the choice model developed in \cref{sec:choice}.
The first question is how much of a subject's preference ranking can be elicited by varying assignment probabilities, to be addressed in the next subsection.
The second question is how much the elicited preference can credibly restrict the deviation directions, discussed in the following subsection.

\subsection{Truthful elicitation of preference}\label{sec:elicitation}

If we introduce the survey stage into the framework of \cref{sec:choice}, it is natural to model the subject as choosing her report to maximize expected utility from the announced lottery. Let
\begin{equation}\label{eq:anticipated}
V_i(z)=\max_{d\in\K}B_{id}-C_{izd},
\qquad C_{izd}>C_{izz}=0\qquad(d\ne z),
\end{equation}
be the value function of initial assignment to \(z\).
We assume that reports do not change perceived treatment benefits or primitive cost parameters. We first consider treatment-choice opportunities that, conditional on assignment, are independent of the report; reports then affect expected utility only through the announced assignment lottery.%
\footnote{The mediated reassignment procedure considered in \cref{sec:feasibility} additionally allows the report to determine the destination of an accommodated request.}

To consider eliciting a ranking, take benefits to be distinct and let \(\pi=(\pi_1,\ldots,\pi_K)\) list the arms in decreasing order of benefit. A feature of this ranking is \emph{strictly elicited} if its true value is the unique optimal report for every admissible benefit-cost profile \((B_{id},C_{izd})\).

The following proposition establishes two things.
First, in the choice model \eqref{eq:choice-model}, we can strictly elicit the favorite arm.
Second, whenever we cannot exclude arm-wise always-takers---as under \(\Ss\) in \cref{thm:impossibility} or even extended monotonicity---we \emph{cannot} strictly elicit the ranking below the favorite.

\begin{proposition}[Truthful preference elicitation]\label[proposition]{prop:truthful}
Suppose the subject maximizes her (expected) value function described in \eqref{eq:anticipated} under the conditional randomization scheme to be specified in each of the following statements.
\renewcommand{\theenumi}{\roman{enumi}}
\renewcommand{\labelenumi}{(\theenumi)}
\begin{enumerate}
\item\label{prop:truthful:i} In the choice model \eqref{eq:choice-model}, fix \(a\in\Delta(\K)\) with \(a_z>0\) for every \(z\), independently of the report, and let \(0<\lambda<1\). Under \eqref{eq:lottery}, the optimal reports are exactly the benefit-maximizing arms. In particular, when the favorite is unique, the assignment rule strictly elicits it.
\item\label{prop:truthful:ii} Let \(K\ge3\), and consider benefit-cost profiles with distinct benefits. Suppose we allow all benefit-cost profiles that induce arm-wise always-takers. Then, no conditional randomization scheme strictly elicits a feature of the ranking not determined by the favorite. Moreover, suppose subjects report a ranking \(\pi\) and are assigned by a lottery \(q(\pi)\in\Delta(\K)\) with \(q_z(\pi)>0\) for every \(z\). Truthful ranking reports are weakly optimal at every admissible profile if and only if
\[
q(\pi)=(1-\lambda)a+\lambda e_{\pi_1}
\]
for some \(a\in\Delta(\K)\) with \(a_z>0\) for every \(z\) and some \(0\le\lambda<1\).
\end{enumerate}
\end{proposition}

The proof is in \cref{app:design}. The second statement of part~(\ref{prop:truthful:ii}) says that the lotteries under which truthful ranking reports are weakly optimal are exactly those of the form \eqref{eq:lottery}, with \(S_i\) equal to the reported first choice and a baseline and boost that do not depend on the report. The order below the favorite does not affect assignment.
The favorite is therefore the most detailed ranking information that these lotteries can elicit credibly and uniformly.
For this reason, we hereafter elicit only the favorite.

We can compare this result with \citet{klm2016}, who build on a similar identification strategy and use the entire ranking.
Their ranking is credible below the top because the admission mechanism is a serial dictatorship, offers are essentially final, and cutoffs are uncertain at application, so every listed course is consequential and the value of an offer is the value of the field \citep[pp.~1072--1074]{klm2016}.
In our setting, subjects can switch after assignment, which is why \cref{prop:truthful} elicits a favorite but no fallback.
\citet{k2023} studies related report-dependent lotteries for binary treatment assignment.

The limit on what an assignment lottery can elicit has counterparts in social choice and in experimental incentives.
A strategy-proof rule from reported rankings to lotteries is a probability mixture of schemes each of which is unilateral or duple \citep[Theorem, p.~671]{gibbard1977}.
Strategy-proofness alone therefore permits lotteries that respond to the entire ranking, such as point voting with decreasing weights, and tops-onlyness for random rules is obtained from strategy-proofness together with unanimity on suitably rich domains \citep[Theorem~1, p.~419]{cz2018}, a combination that for a single agent means assigning the reported favorite with certainty.
The characterization above is instead a random dictatorship mixed with a constant lottery, obtained without unanimity or efficiency; the constant part is deliberate, because every arm must remain reachable.
In the vocabulary of \citet[Theorem~1, p.~1485]{ach2018}, the announced lottery is a random set selection mechanism, and for one menu of \(K\) arms their strictly incentive compatible mechanisms take the same baseline-plus-increment form.
What is specific to our setting is why nothing beyond the favorite can be elicited.
Subjects with the same favorite can rank the other arms differently while obtaining the same expected payoff. Consequently, no assignment lottery can strictly elicit the lower ranking for all admissible benefit and cost profiles. Report-dependent transfers that enter utility additively do not change this conclusion.
This differs from constrained elicitation, where lottery preferences are the ranking itself, the entire ranking is elicitable, and only coarse summaries such as the second-ranked alternative alone are not \citep[Section~5, p.~526]{ach2021}.

\subsection{Feasibility of compliance to stated preference}\label{sec:feasibility}

Now given the truthful revelation of the favorite, how do we know that deviation only takes place toward it?

First, consider the case where compliance is reasonably postulated to be strictly stable.
\citet[Section~7]{bht2025} report that strict stability is not rejected in the French counseling experiment of \citet{bcg2014}.
\citet{hhklm2026} assume stability in the context of the admission mechanism and field-specific offers, as discussed in \cref{sec:stable}.
More generally, stability is plausible when assignment increases the attractiveness of its targeted treatment while leaving the relative appeal of other treatments unchanged. With a unique preferred treatment in the absence of encouragement, an individual either takes the assigned treatment or returns to that same baseline choice. \citet{fgmp2016} provide a substantive argument of this kind. Because a Head Start offer leaves the availability of other childcare options unchanged, it should not change the choice among those options. Related irrelevance restrictions appear in the admissions settings studied by \citet{drs2023} and \citet{altmejd2026}, where an offer that does not induce participation in the offered field is assumed not to change the choice among other fields. \citet{hmr2025} provide related motivation for encouragement effects confined to their targeted interventions, arguing that classroom and on-the-job training serve distinct purposes and documenting an absence of cross-effects from program-specific instruments. These arguments motivate (strict) stability in applications where encouragement lowers the bar to take the assigned treatment and preserves an individual's baseline ranking of the remaining alternatives. In our RCT setup, this corresponds to deviation costs that do not depend on assignment.

In the destination-dependent cost model, any deviation leads to an arm maximizing the baseline net value \(B_{id}-C_{id}\), whereas the lottery elicits an arm maximizing \(B_{id}\). For a subject who deviates, compliance to stated preference therefore requires the reported favorite \(S_i\) to maximize baseline net value as well; this is more plausible when the variation of deviation costs is not too large.
If the baseline net-value maximizer is unique and equals \(S_i\), every deviation leads to \(S_i\). If several arms maximize baseline net value and \(S_i\) is among them, the tie-breaking rule for deviations must favor \(S_i\).
In the choice model, a sufficient but not necessary configuration is when the deviation cost is uniform across arms. In \eqref{eq:stable-model}, impose further that
\[
R_i(z)\in\argmax_{d\in\mathcal{K}}B_{id}-C_i\ind{d\ne z},\qquad C_i>0.
\]
Then, with \(S_i\in\argmax_d B_{id}\) uniquely, we have \(R_i(z)\in\{z,S_i\}\) for every \(z\).

Next, consider the case where we want to accommodate the possibility that the deviation costs might depend on assignment.
When deviation costs depend on assignment, the feasibility of compliance to stated preference depends on the administrator's ability to mediate treatment switching. If switching requires authorization and the reported arm is available, then after some negotiation attempts by the subject, the administrator can \emph{reassign} her to that arm.
Indeed, the choice model \eqref{eq:choice-model} guarantees that there exists a cost of negotiation \(c_{iz}\) such that reassignment is strictly preferable for a subject who would otherwise deviate.
Namely, take \(0<c_{iz}<\min_{d\ne z}C_{izd}\).
Suppose the subject would choose \(d\ne z\) in the original choice problem. Since \(S_i\) is her truthfully reported favorite,
\[
B_{iS_i}-c_{iz}>B_{id}-C_{izd}\ge B_{iz}.
\]
She therefore strictly prefers reassignment over any other arm.
Thus, reassignment controls deviation even when the original compliance function is only acyclic.
This \(c_{iz}\) may come from the burden to negotiate or also be embedded more explicitly as an opportunity cost, such as the monetary incentive provided only upon her compliance.

Truthful reporting also remains strictly optimal even when subjects know this procedure, provided \(\lambda>0\) and the negotiation cost \(c_{iz}>0\) is independent of the report for every initial assignment \(z\). Under report \(s\), the continuation value is
\[
\tilde{V}_i(z;s)=\max\{B_{iz},B_{is}-c_{iz}\}.
\]
If \(t\) is the unique favorite, then \(\tilde{V}_i(z;t)\ge \tilde{V}_i(z;s)\) for every \(z\), while \(\tilde{V}_i(t;t)=B_{it}\) and \(\tilde{V}_i(s;s)\le V_i(s)\). Thus expected utility \(U_i(s)\) from reporting \(s\) under \eqref{eq:lottery} satisfies
\[
U_i(t)-U_i(s)\ge\lambda(B_{it}-V_i(s))>0\qquad(s\ne t),
\]
so reporting the favorite remains strictly optimal.

Such administrative control over program entry appears in applications such as the French counseling experiment, whose protocol redirected refusers to standard counseling and denied access to the other intensive program \citep[p.~148]{bcg2014}.
Even when all switching goes through the administrator, redirecting every refuser to one fixed arm requires staff to enforce the protocol. In that experiment, contractors pressed staff to send them unassigned job seekers \citep[p.~155]{bcg2014}. The reported favorite asks staff only to grant what the subject wants, and she has no reason to leave it, so compliance to stated preference is easier to enforce. Report-specific destinations also identify comparisons between active arms, whereas a fixed destination identifies each arm only against that destination.

Meanwhile, such control is harder to exercise when equivalent treatments are available outside the supervision of the experimental administrator. The controls in the JTPA experiment frequently obtained comparable training through other providers or funding sources \citep[pp.~652, 661--662]{hhsk2000}, and Head Start controls could enroll in other preschool programs \citep[p.~1254]{fgmp2016}. In these settings, controlling reassignment within the experiment does not ensure \cref{ass:preference}. Compliance to stated preference must therefore rest on an additional behavioral restriction, such as strict stability with its common deviation destination equal to the reported favorite. The feasibility of such restrictions can then be tested by the aforementioned testable implication, \(D_i\in\{Z_i,S_i\}\).

All of the above arguments presume that the reported arm remains a favorite between the survey stage and the treatment stage. The required dynamic consistency of preferences is plausible when the survey and treatment occur close together in time.

\section{Conclusion}\label{sec:conclusion}

Multi-arm experiments allow researchers to compare several policies in one study, but imperfect compliance can make comparisons between received treatments refer to incompatible groups. With three or more arms, even the intersection of ordered monotonicity, strict stability, and unordered monotonicity is not strong enough to isolate any group between any two treatment arms. If outcomes are unrestricted, these compliance restrictions do not identify nontrivial margin-specific group ATEs for any groups, including convex combinations of LATEs for the same treatment margin.

When compliance decisions come from economic choice behavior, our choice model points to a design that permits identification of margin-specific effects. While rational choice with positive deviation costs gives rise to a broad class of acyclic compliance, it also hints at the incentive structure that the experimenter can exploit. A lottery tilted toward the reported arm elicits a favorite.
And when treatment receipt is confined to the assigned arm or the reported favorite, conditioning on the report restores one-sided noncompliance. Then, an appropriate Wald ratio identifies the favorite's effect relative to each alternative for subjects who would comply with the alternative.

The feasibility of compliance to stated preference depends on the administrator's control over treatment switching. When switching requires authorization, reassignment to a truthfully reported favorite removes the incentive to deviate further.
This scenario can accommodate acyclic compliance in general.
When subjects can obtain equivalent treatments elsewhere, the administrator may not be able to practically enforce compliance to stated preference. The restriction must instead rest on a behavioral assumption, such as strict stability with the reported favorite as the common deviation destination. Its observable implication can be tested after the experiment.

Broadly speaking, this paper sheds light on the role of compliance assumptions in experimental design. Preference elicitation, random assignment, and the handling of deviations can be designed together to support identification of treatment effects.

\begin{appendices}
\appendix
\appendixpage

All proofs in this paper are verified in Lean.%
\footnote{The code is available at \url{https://github.com/kajitetsuya/lean-late-design}.}

\section{Group alignment and nonidentification}\label[appendix]{app:impossibility}

\begin{proof}[Proof of \cref{lem:alignment}]
First characterize the nonidentity members of the class. Strict stability gives one common destination \(q\) for all nonfixed assignments. The destination must be fixed, for otherwise its own deviation would also have to equal \(q\). Therefore the map leaves every point outside the basin \(A_q(R)\) fixed and sends every point in it to \(q\). Since \(R\) is ordered monotone, \(A_q(R)\) is an interval containing \(q\). Conversely, every interval collapse in \eqref{eq:collapse} is strictly stable and ordered monotone.

For each \(q\), let \(\mathcal C_q\) consist of the intervals \(A\) for which \(M_{q,A}\in\Ss\); note that \(M_{q,\{q\}}=I\). A collapse with destination \(q\) has \(q\)-basin \(A\); every other treatment \(d\) has basin either \(\{d\}\) or \(\varnothing\). Thus unordered monotonicity is equivalent, in this class, to requiring that each \(\mathcal C_q\) be a chain. There are no further compatibility constraints between different destinations.

Maximality of \(\Ss\) now implies that each \(\mathcal C_q\) is a maximal chain among the intervals containing \(q\). In particular, \(\{q\}\) and \(\K\) belong to \(\mathcal C_q\), since each is comparable with every interval containing \(q\). If successive intervals differed by at least two assignments, an intermediate interval could be added by extending one endpoint, contradicting maximality. Hence
\begin{equation}\label{eq:saturated}
A_q^1=\{q\}\subsetneq A_q^2\subsetneq\cdots
\subsetneq A_q^K=\K,\qquad \lvert A_q^m\rvert=m.
\end{equation}

Fix arbitrary \(t\ne s\) and coefficient vectors satisfying \eqref{eq:alignment}. The identity gives \(b_t=c_s\eqqcolon\alpha\). For every \(A\in\mathcal C_t\), the collapse with destination \(t\) yields
\[
\sum_{z\in A}b_z=\alpha\ind{s\notin A}.
\]
Take successive differences along \eqref{eq:saturated}. The new coefficient is zero except at the step adding \(s\), where it equals \(-\alpha\). This proves \(b=\alpha(e_t-e_s)\). Similarly, the \(s\)-chain gives \(\alpha\ind{t\notin A}=\sum_{z\in A}c_z\), and successive differences give \(c=\alpha(e_s-e_t)\).

Choose \(q\notin\{t,s\}\). Its chain begins at a set containing neither \(t\) nor \(s\) and ends at a set containing both. At the first step including either, exactly one has been added; let \(A\) be this interval. If \(A\) contains \(t\) but not \(s\), the \(t\)-basin of \(M_{q,A}\) is empty and its \(s\)-basin is \(\{s\}\), so \eqref{eq:alignment} reads \(0=c_s=\alpha\). If \(A\) contains \(s\) but not \(t\), it reads \(b_t=0\), so again \(\alpha=0\). Thus \(b=c=0\).
\end{proof}

\begin{proof}[Proof of \cref{thm:impossibility}]
Fix \(t\ne s\). The proof has five steps. Step~1 restates \cref{lem:alignment} as a statement about row spaces. Step~2 uses it to show that a linear functional of outcomes is identified only if it vanishes. Steps~3--5 apply this claim to the three kinds of target in the theorem: groups given as distributions over subjects, groups defined relative to the population, and outcome-defined subpopulations.

\textit{Step 1 (row spaces).} We work with a finite collection of latent cells \(j\in\mathcal J\). Each cell has a compliance function \(r_j\). Several cells may share a compliance function, but every compliance function in \(\Rs\) must belong to at least one cell. Let \(J_i\) denote the cell of subject \(i\), and write \(\mathcal I_j=\{i\in\mathcal I:J_i=j\}\). Define
\[
[\Gamma_d]_{zj}=\ind{r_j(z)=d},\qquad
\Gamma_d\in\mathbb R^{K\times\lvert\mathcal J\rvert},
\]
so that the \(j\)th column of \(\Gamma_d\) is the indicator of the basin \(A_d(r_j)\). The \(j\)th coordinate of \(\Gamma_t^\top b\) is \(\sum_{z:r_j(z)=t}b_z\). Hence the equality \(\Gamma_t^\top b=\Gamma_s^\top c\) is \eqref{eq:alignment} for every \(R\in\Rs\), and in particular for every \(R\in\Ss\). By \cref{lem:alignment}, \(b=c=0\). Therefore
\begin{equation}\label{eq:zero-intersection}
\row(\Gamma_t)\cap\row(\Gamma_s)=\{0\}.
\end{equation}

\textit{Step 2 (an outcome perturbation).}
\begin{claim}
Fix positive cell probabilities \(p_j=P(J_i=j)\), and let \(\phi_j\), \(j\in\mathcal J\), be measurable real-valued functions with \(\phi_j(y)=0\) whenever \(y_t=y_s\). Consider the linear functional
\[
L(P)=\sum_{j\in\mathcal J}p_j\E_P[\phi_j(\mathbf Y_i)\mid J_i=j].
\]
If \(L\) is identified over the populations that share these cell probabilities and differ only in the conditional outcome laws given the cell, then every \(\phi_j\) is identically zero.
\end{claim}

To prove the claim, fix an outcome vector \(y\) with \(y_t\ne y_s\), put \(g_j=\phi_j(y)\), and let \(y^{(t)}\) be \(y\) with its \(t\)th coordinate replaced by \(y_s\). In every cell, let \(\mathbf Y_i\) equal \(y\) or \(y^{(t)}\) with probability \(1/2\) each. For any \(v\in\ker\Gamma_t\), change these probabilities in cell \(j\) to
\[
\frac12+\epsilon\frac{v_j}{p_j}
\qquad\text{and}\qquad
\frac12-\epsilon\frac{v_j}{p_j},
\]
respectively, with \(\epsilon\ne0\) small enough to keep them positive. Both populations have two-point outcome distributions, so all moments are finite, and both use the same independent assignment rule. Conditional on assignment \(z\), the probabilities of \((Y,D)=(y_t,t)\) and \((Y,D)=(y_s,t)\) change by \(\epsilon(\Gamma_tv)_z=0\) and its negative, and all other observable probabilities are unchanged. Thus the observable law is unchanged, while \(L\) changes by \(\epsilon g^\top v\). Identification therefore requires \(g\perp\ker\Gamma_t\), that is, \(g\in\row(\Gamma_t)\). Repeating the argument with the \(s\)th coordinate of \(y\) replaced by \(y_t\) gives \(g\in\row(\Gamma_s)\), and \eqref{eq:zero-intersection} gives \(g=0\). Since \(y\) was arbitrary and each \(\phi_j\) vanishes when \(y_t=y_s\), the claim follows.

\textit{Step 3 (groups given as distributions over subjects).} Let \(G_1,\ldots,G_m\) be groups specified without potential outcomes, and let \(a_1,\ldots,a_m\) be real coefficients with \(\nu=\sum_{\ell=1}^m a_\ell G_\ell\ne0\); if \(\nu=0\), the contrast is identically zero. Choose the population distribution as
\[
\mu=\frac{1}{2}\frac{\lvert\nu\rvert}{\lvert\nu\rvert(\mathcal I)}+\frac{1}{2}H,
\]
where \(\lvert\nu\rvert\) is the total variation measure of \(\nu\) and \(H\) is a probability measure on \(\mathcal I\) that places positive mass on one subject index of every allowed compliance type. Then \(\nu\) is absolutely continuous with respect to \(\mu\), with bounded density \(f=d\nu/d\mu\). Form cells by partitioning subjects according to compliance type and to whether \(f\) is positive, negative, or zero, and omit cells with zero \(\mu\)-probability. A compliance type may thus occupy up to three cells, which Step~1 allows. Put \(p_j=\mu(\mathcal I_j)>0\) and \(h_j=\nu(\mathcal I_j)/p_j\). Every compliance type is represented. Moreover, \(h\ne0\). Since \(f\) has a constant sign on each cell, \(h=0\) would force \(f=0\) \(\mu\)-almost everywhere, and hence \(\nu=0\).

Restrict conditional outcome laws to depend only on the cell, and set outcomes to zero on \(\mu\)-null indices. Then
\[
\sum_{\ell=1}^m a_\ell\tau_{G_\ell}(t,s)
=\sum_j p_jh_j\E[Y_{it}-Y_{is}\mid J_i=j].
\]
This is \(L\) with \(\phi_j(y)=h_j(y_t-y_s)\). These functions do not all vanish, so the claim gives nonidentification. The construction holds \(\mu\), all pre-assignment characteristics, and all \(G_\ell\) fixed. It proves \cref{rem:signed-groups}, and taking \(m=1\), \(a_1=1\) gives the result for a single group ATE.

\textit{Step 4 (groups defined relative to the population).} Groups formed by conditioning on pre-assignment characteristics, or by normalizing fixed finite nonnegative functions of those characteristics, depend on \(\mu\). We therefore choose \(\mu\) first and then form the groups. Take \(\mu\) finitely supported, with positive probability on every compliance type and with the required normalizing constants finite and positive. Each resulting \(G_\ell\) is absolutely continuous with respect to \(\mu\), and so is \(\nu=\sum_{\ell}a_\ell G_\ell\). If \(\nu\ne0\), the cells and the outcome construction of Step~3 apply without change. Because the groups do not depend on outcomes, \(\mu\) and every \(G_\ell\) stay fixed while outcomes change.

\textit{Step 5 (outcome-defined subpopulations).} Now consider a fixed measurable membership event \(A\) in \((X,\mathbf Y)\). The criterion \(A\) stays fixed below, although which subjects satisfy it changes with potential outcomes. For a population law \(P\) with \(P(A)>0\), write
\[
N(P)=\E_P[\ind{A}(Y_t-Y_s)],\qquad
\tau_A(P)=\frac{N(P)}{P(A)}.
\]
Suppose \(\tau_A\) is identified and nontrivial. Take any two observationally equivalent laws \(P,P'\) and any \(Q\) with \(Q(A)>0\). For \(0<\epsilon<1\), the mixtures \(P_\epsilon=(1-\epsilon)Q+\epsilon P\) and \(P'_\epsilon=(1-\epsilon)Q+\epsilon P'\) are admissible and observationally equivalent, and both give \(A\) positive probability, so \(\tau_A(P_\epsilon)=\tau_A(P'_\epsilon)\). Because \(N\) and the probability of \(A\) are linear in the law, and \(N(Q)=\tau_A(Q)Q(A)\),
\[
\tau_A(P_\epsilon)-\tau_A(Q)
=\frac{N(P_\epsilon)-\tau_A(Q)P_\epsilon(A)}{P_\epsilon(A)}
=\frac{\epsilon\,[N(P)-\tau_A(Q)P(A)]}{P_\epsilon(A)},
\]
and likewise for \(P'_\epsilon\). Equate the two expressions, divide by \(\epsilon\), and let \(\epsilon\downarrow0\), so that both denominators tend to \(Q(A)\). This gives
\[
N(P)-N(P')=\tau_A(Q)\,[P(A)-P'(A)].
\]
Since \(\tau_A\) is nontrivial, two choices of \(Q\) give different values of \(\tau_A(Q)\), so both differences are zero. Thus \(N\) and \(P(A)\) are separately identified, including at laws under which \(A\) has zero probability.

To apply the claim to \(N\), consider populations in which \(X\) takes one allowed value \(x_j\) for each compliance function in \(\Rs\), with \(r_j\) the compliance function recorded in \(x_j\), and fixed positive probabilities \(p_j\); cell \(j\) consists of the subjects with \(X=x_j\). Then \(N=L\) with \(\phi_j(y)=\ind{(x_j,y)\in A}(y_t-y_s)\), so the claim makes every \(\phi_j\) vanish. Because the representatives \(x_j\) can be chosen arbitrarily, \(\ind{A}(Y_t-Y_s)\) vanishes at every allowed value of \((X,\mathbf Y)\). Hence \(\tau_A(P)=0\) whenever it is defined, contradicting nontriviality.
\end{proof}

\section{Choice-theoretic representations}\label[appendix]{app:representations}

\begin{proof}[Proof of \cref{thm:acyclic}]
If \(R(z)\ne z\) maximizes \(B_d-C_{zd}\), optimality gives \(B_{R(z)}-C_{z,R(z)}\ge B_z\). Since \(C_{z,R(z)}>0\), we have \(B_{R(z)}>B_z\). A nontrivial cycle would produce a strict chain of benefit inequalities returning to its starting value, which is impossible.

Conversely, suppose \(R\) has no nontrivial cycle. Since \(\K\) is finite, every orbit ends at a fixed point. Let \(h(z)\) be the number of moves needed to reach a fixed point, with \(h(z)=0\) at a fixed point, and set \(B_z=-h(z)\). At a nonfixed point, \(h(R(z))=h(z)-1\), so \(B_{R(z)}=B_z+1\). Define
\[
C_{zd}=
\begin{cases}
0,&d=z,\\
1/2,&d=R(z)\ne z,\\
K+1,&\text{otherwise}.
\end{cases}
\]
Every off-diagonal cost is positive. At a nonfixed assignment, the prescribed destination gives \(B_z+1/2\), strictly more than compliance. All other unassigned destinations give at most \(-K-1\), whereas \(B_z\ge-(K-1)\). At a fixed assignment, compliance gives zero and every other destination has negative net benefit. The maximizer is therefore uniquely \(R(z)\) for every \(z\).
\end{proof}

\begin{proof}[Proof of \cref{thm:ordered}]
For necessity, let \(R\) select a maximizer at each assignment, with ties within a row broken arbitrarily, and write \(Q_{zd}=B_d-C_{zd}\). Strict submodularity of \(C\) means that \(Q\) has strictly increasing differences:
\[
Q_{za}+Q_{z'b}>Q_{zb}+Q_{z'a}
\qquad(z<z',\ a<b).
\]
If \(R\) decreased, say \(R(z)=b>a=R(z')\), optimality would give \(Q_{zb}\ge Q_{za}\) and \(Q_{z'a}\ge Q_{z'b}\), contradicting the displayed strict inequality.

For sufficiency, let \(R\) be ordered monotone and choose \(L>K\). Define
\[
Q_{zd}=-L(d-R(z))^2+zd.
\]
If \(r=R(z)\) and \(d\ne r\), then
\[
Q_{zr}-Q_{zd}
=L(d-r)^2-z(d-r)>0.
\]
For \(d>r\), the difference is \((d-r)[L(d-r)-z]>0\); for \(d<r\), both terms in the preceding expression are positive. Hence \(R(z)\) uniquely maximizes row \(z\). Moreover, for \(z<z'\) and \(d<d'\),
\[
Q_{zd}+Q_{z'd'}-Q_{zd'}-Q_{z'd}
=(d'-d)\{2L[R(z')-R(z)]+z'-z\}>0.
\]

It remains to obtain normalized positive deviation costs. For \(j=1,\ldots,K-1\), let \(\Delta_j(z)=Q_{z,j+1}-Q_{zj}\). Strictly increasing differences imply that \(\Delta_j(z)\) is strictly increasing in \(z\). Set \(B_1=0\) and choose
\begin{equation}\label{eq:ordered-normalization}
\Delta_j(j)<B_{j+1}-B_j<\Delta_j(j+1).
\end{equation}
These intervals are nonempty. Define
\[
C_{zd}=B_d-B_z-Q_{zd}+Q_{zz}.
\]
Then \(C_{zz}=0\), and \(B_d-C_{zd}=B_z+Q_{zd}-Q_{zz}\), so the maximizers are unchanged.

If \(d>z\), summing \eqref{eq:ordered-normalization} from \(j=z\) to \(d-1\) gives
\[
B_d-B_z>\sum_{j=z}^{d-1}\Delta_j(z)=Q_{zd}-Q_{zz},
\]
because \(z\le j\). If \(d<z\), the upper inequalities in \eqref{eq:ordered-normalization}, summed from \(j=d\) to \(z-1\), give
\[
B_z-B_d<\sum_{j=d}^{z-1}\Delta_j(z)=Q_{zz}-Q_{zd},
\]
because \(j+1\le z\). Thus \(C_{zd}>0\) in either case.

Adding a row term and a column term does not change a cross difference, so
\[
C_{zd}+C_{z'd'}-C_{zd'}-C_{z'd}
=-\{Q_{zd}+Q_{z'd'}-Q_{zd'}-Q_{z'd}\}<0.
\]
The constructed costs satisfy all requirements.
\end{proof}

\begin{proof}[Proof of \cref{thm:stable}]
Suppress the individual index and write \(v_d=B_d-C_d\) and \(M=\argmax_d v_d\). Assignment to \(z\) increases only the value of \(z\), from \(v_z\) to \(v_z+C_z\).
At every assignment \(z\in M\), the subject strictly prefers to comply, since \(C_z>0\). Thus \(R(z)=z\) on \(M\). If a subject assigned outside \(M\) deviates, the chosen destination must belong to \(M\) because any other unassigned destination has value below a member of \(M\). Every selected destination is therefore a fixed point, proving stability.

Conversely, let \(R\) be stable, and let \(F=\{z:R(z)=z\}\). Set \(C_d=1\) for all \(d\), \(v_d=2\) for \(d\in F\), and \(v_d=0\) otherwise, with \(B_d=v_d+C_d\). An assignment in \(F\) has unique value 3. An assignment outside \(F\) has value 1, while all arms in \(F\) have value 2. Stability ensures \(R(z)\in F\), so a deterministic tie rule choosing this destination represents \(R\).

For the unique-maximizer claim, suppose a nonidentity compliance function is generated without ties at the maximum. At any nonfixed assignment outside \(M\), every member of \(M\) remains available at the common maximal value. Thus \(M\) must be a singleton whenever a deviation occurs. All deviations then go to that singleton, giving strict stability.

For the converse, first consider the identity. Take \(v_d=0\) and \(C_d=1\) for every \(d\), so each assignment uniquely selects itself. For a nonidentity strictly stable map, let \(q\) be its common deviation destination. This destination is fixed. Set \(C_d=1\) for all \(d\), and
\[
v_q=1,\qquad
v_d=\tfrac{1}{2}\ \text{for other fixed points},\qquad
v_d=-1\ \text{for nonfixed points}.
\]
With \(B_d=v_d+1\), assignment to a fixed point uniquely selects that point, while assignment to a nonfixed point uniquely selects \(q\). This represents the prescribed function.
\end{proof}

\begin{proof}[Proof of \cref{thm:unordered}]
For necessity, fix a representing pair \((C,\mathcal B)\). Admissibility means that, for each assignment \(z\), exactly one treatment \(d\) has \(B_d-C_{zd}>0\), while all others have nonpositive net payoff. It follows that
\[
A_d(R_B)=\{z:C_{zd}<B_d\}=T_d(B).
\]
For a fixed \(d\), these threshold sets are nested as the scalar \(B_d\) varies. Hence the generated collection is unordered monotone. The included profile \(B^I\) has \(T_d(B^I)=\{d\}\) and generates the identity.

Conversely, let \(\Rs\) be unordered monotone and contain \(I\). For each destination \(d\), its basin collection is a finite chain containing \(\{d\}\). Every nonempty set in the chain contains \(d\). We may therefore order the assignments by a permutation \(\sigma_d\), with \(\sigma_d(1)=d\), so that every basin is a prefix:
\[
A_d(R)=\{\sigma_d(j):1\le j\le \lvert A_d(R)\rvert\}.
\]
An empty basin corresponds to the empty prefix.

Define the common cost matrix and benefit profiles by
\[
C_{\sigma_d(j),d}=j-1\quad(j=1,\ldots,K),\qquad B_d(R)=\lvert A_d(R)\rvert.
\]
The costs satisfy \(C_{dd}=0<C_{zd}\) off the diagonal, and \(C_{zd}<B_d(R)\) exactly when \(z\in A_d(R)\). Since the basins of \(R\) partition \(\K\), the benefit profile \(B(R)\) is admissible and uniquely generates \(R\). Choose
\(\mathcal B=\{B(R):R\in\Rs\}\). It includes \(B(I)=(1,\ldots,1)\), which satisfies the requirement on \(B^I\). The generated collection is exactly \(\Rs\).
\end{proof}

\section{Possibility and truthful revelation}\label[appendix]{app:design}

\begin{proof}[Proof of \cref{thm:possibility}]
Condition on \(S_i=d\). Assignment to \(d\) produces \(Y_{id}\). Under assignment to \(z\), subjects with \(R_i(z)=z\) produce \(Y_{iz}\), while the others produce \(Y_{id}\). By conditional randomization, the difference in means is
\[
\E[\ind{R_i(z)=z}(Y_{id}-Y_{iz})\mid S_i=d].
\]
The take-up rate is \(P(R_i(z)=z\mid S_i=d)\). Dividing by it gives the result.
\end{proof}

\begin{proof}[Proof of \cref{prop:truthful}]
For part~(\ref{prop:truthful:i}), expected perceived utility from reporting \(s\) is
\[
U_i(s)=(1-\lambda)\sum_{z\in\K}a_zV_i(z)+\lambda V_i(s).
\]
For a favorite \(t\), \(V_i(t)=B_{it}=\max_d B_{id}\). At any nonfavorite assignment \(z\), taking \(z\) gives less than this maximum, and every alternative incurs a strictly positive cost. Since there are finitely many arms, \(V_i(z)<B_{it}\). Therefore reporting a favorite maximizes \(U_i(s)\), uniquely when the favorite is unique.

For part~(\ref{prop:truthful:ii}), we construct admissible profiles under which the subject takes her favorite \(t\) under every assignment. Suppress the subject index. Fix a vector \(v\in\mathbb R^K\) uniquely maximized at \(t\), and any ranking \(\pi\) with \(\pi_1=t\). Set
\[
C_{zz}=0,\qquad C_{zd}=(v_t-v_z)+(v_t-v_d)\quad(z\ne d).
\]
Every off-diagonal cost is positive. Choose \(B_t=v_t\) and distinct \(B_z<\min_d v_d\) for \(z\ne t\), ordered according to \(\pi\). Treatment \(t\) then uniquely maximizes net benefit under every assignment, giving
\[
R(z)=t\ \text{for every } z,\qquad V(z)=B_t-C_{zt}=v_z.
\]
Thus every ranking with favorite \(t\) can be realized with the same costs, compliance function, and assignment values. If a ranking feature differs between two such rankings, the corresponding profiles have different truthful reports but evaluate every report's lottery identically. Their truthful reports cannot both be uniquely optimal. This proves the first statement of part~(\ref{prop:truthful:ii}), with unique post-assignment choices in every constructed profile.

For the second statement, let \(q(\pi)\) be a lottery as in the statement for which truthful ranking reports are weakly optimal at every admissible profile. For two rankings \(\pi\) and \(\pi'\) with favorite \(t\), the incentive constraints at the two corresponding profiles give
\[
[q(\pi)-q(\pi')]\cdot v=0
\quad\text{for every }v\text{ uniquely maximized at }t.
\]
This set of vectors is open in \(\mathbb R^K\), so \(q(\pi)=q(\pi')\). Write the common lottery as \(q^t\).

For distinct favorites \(t\) and \(t'\), incentive compatibility requires
\[
(q^t-q^{t'})\cdot v\ge0
\quad\text{whenever }v_t>\max_{z\ne t}v_z.
\]
Letting any coordinate \(v_z\), \(z\ne t\), tend to minus infinity shows that \((q^t-q^{t'})_z\le0\). Applying the reverse constraint on the domain with favorite \(t'\) gives \((q^t-q^{t'})_z\ge0\) for \(z\ne t'\). Since lottery differences sum to zero,
\[
q^t-q^{t'}=\lambda_{tt'}(e_t-e_{t'}),\qquad \lambda_{tt'}\ge0.
\]

For three distinct favorites \(t\), \(t'\), and \(t''\), the identity \(q^t-q^{t''}=(q^t-q^{t'})+(q^{t'}-q^{t''})\) forces all three coefficients to agree. Hence there is a common \(\lambda\ge0\) for which \(q^t-\lambda e_t\) is independent of \(t\). Every coordinate of this vector is positive: for coordinate \(z\), choose \(t\ne z\), so it equals \(q_z^t>0\). Its coordinates sum to \(1-\lambda\), so \(\lambda<1\). Define \(a=(q^t-\lambda e_t)/(1-\lambda)\). Then
\[
q(\pi)=(1-\lambda)a+\lambda e_{\pi_1},\qquad
a_z>0,\qquad \sum_{z\in\K}a_z=1.
\]

Conversely, under any profile in \eqref{eq:anticipated} with distinct benefits, part~(\ref{prop:truthful:i}) shows that the favorite uniquely maximizes assignment value. Expected utility under this lottery is \((1-\lambda)a\cdot V+\lambda V(\pi_1)\), so truthful ranking reports are weakly optimal.
\end{proof}

\end{appendices}

\begingroup
\singlespacing
\small
\setlength{\bibsep}{2pt}
\bibliographystyle{econ}
\bibliography{design}
\endgroup
\end{document}